\documentclass[journal,twoside]{IEEEtran}
\usepackage{amssymb}
\usepackage{amsmath}
\usepackage{amsthm}
\usepackage{graphicx}
\usepackage{algorithm}
\usepackage{algorithmic}
\usepackage{cite}
\usepackage{bm}
\usepackage{float}
\usepackage{multirow}
\usepackage{color}
\usepackage{subfigure}
\usepackage{epstopdf}
\usepackage{hyperref}
\usepackage{pgfplots}
\usepackage{pgfplotstable}
\usepackage{enumitem}

\begin{document}

\title{Economics of Semantic Communication System: An Auction Approach}
%
\author{Zi Qin Liew, Hongyang Du, Wei Yang Bryan Lim, Zehui Xiong, Dusit Niyato, \IEEEmembership{Fellow, IEEE}, Chunyan~Miao,~\IEEEmembership{Senior~Member,~IEEE}, and Dong In Kim,~\IEEEmembership{Fellow, IEEE} 

\thanks{Zi Qin Liew and Wei Yang Bryan Lim are with Alibaba-NTU Singapore Joint Research Institute, Nanyang Technological Univerity, Singapore (e-mail: ziqin001@e.ntu.edu.sg, limw0201@e.ntu.edu.sg).}
\thanks{H.~Du is with the School of Computer Science and Engineering, the Energy Research Institute @ NTU, Interdisciplinary Graduate Program, Nanyang Technological University, Singapore (e-mail: hongyang001@e.ntu.edu.sg).}
\thanks{Z. Xiong is with the Pillar of Information Systems Technology and Design, Singapore University of Technology and Design, Singapore (e-mail: zehui\_xiong@sutd.edu.sg)}
\thanks{Dusit Niyato is with the School of Computer Science and Engineering, Nanyang Technological University, Singapore (e-mail: dniyato@ntu.edu.sg).} 
\thanks{Chunyan Miao is with the School of Computer Science and Engineering, Nanyang Technological University, Singapore 639798, and also with the Joint NTU-UBC Research Centre of Excellence in Active Living for the Elderly (LILY), Nanyang Technological University, Singapore 639798 (e-mail: ascymiao@ntu.edu.sg).}
\thanks{Dong In Kim is with the Department of Electrical and Computer Engineering, Sungkyunkwan University, Suwon 16419, South Korea (email:dikim@skku.ac.kr).}}


\maketitle

\begin{abstract}
Semantic communication technologies enable wireless edge devices to communicate effectively by transmitting semantic meaning of data. Edge components, such as vehicles in next-generation intelligent transport systems, use well-trained semantic models to encode and decode semantic information extracted from raw and sensor data.
 However, the limitation in computing resources makes it difficult to support the training process of accurate semantic models on edge devices. As such, edge devices can buy the pretrained semantic models from semantic model providers, which is called ``semantic model trading". Upon collecting semantic information with the semantic models, the edge devices can then sell the extracted semantic information, e.g., information about urban road conditions or traffic signs, to the interested buyers for profit, which is called ``semantic information trading". To facilitate both types of the trades, effective incentive mechanisms should be designed. Thus, in this paper, we propose a hierarchical trading system to support both semantic model trading and semantic information trading jointly. The proposed incentive mechanism helps to maximize the revenue of semantic model providers in the semantic model trading, and effectively incentivizes model providers to participate in the development of semantic communication systems. For semantic information trading, our designed auction approach can support the trading between multiple semantic information sellers and buyers, while ensuring individual rationality, incentive compatibility, and budget balance, and moreover, allowing them achieve higher utilities than the baseline method.

\end{abstract}

\section{Introduction}
\label{sec:intro}
With the advancement of sixth-generation (6G) mobile communication technology, data transmission rate in the conventional communication systems is increasing but approaching the Shannon limit. Meanwhile, the remaining available spectrum resources are becoming increasingly scarce. To solve this dilemma, semantic communication technologies are proposed~\cite{xie2020lite}, which aims to transmit the extracted semantic information relevant to the communications goal. Because the data amount that needs to be transmitted can be reduced significantly while ensuring the effectiveness of communications~\cite{xie2021deep}, semantic communications can be widely used in intelligent wireless networks, to enable smart transportation \cite{lin2020edge}, smart logistic \cite{song2020applications}, smart cities \cite{an2019toward}, smart homes \cite{sezer2015development}, and smart healthcare \cite{thangaraj2016agent}.



Existing semantic communication systems \cite{xie2021deep}, \cite{weng2021semantic} are pretrained with labeled datasets with certain channel models. However, a main drawback is that the accuracy and performance of a pretrained semantic model decrease when the background knowledge or communication environment changes, i.e., mismatch between the knowledge base/channel model used in the training and the actual knowledge base/channel model. To reduce the gap in performance, fine-tuning of the model parameters can be done based on the real channel models \cite{dorner2017deep} and new background knowledge \cite{xie2021deep}. However, edge and Internet of Things (IoT) devices with limited computation power might not have enough resources for fine-tuning. Moreover, the results of fine-tuning depend highly on the amount of labeled data of the new knowledge base. To solve the aforementioned problems, inspired by the model trading framework in collaborative edge learning \cite{lim2020federated}, we can adopt a trading system in which model providers trade the trained model to other devices. Specifically, the semantic model provider has more resources to train quality semantic models with the relevant knowledge base and channel models, and the edge devices can obtain the semantic model (semantic encoder/decoder) from the model providers. Using the semantic model, edge devices can extract semantic information from the collected raw data. This enables semantic information exchange between edge devices. Furthermore, as the semantic information is helpful for the decision making of smart agents~\cite{yun2021attention}, the trading of semantic information should also be studied. Using the semantic models, the edge devices collect and trade the semantic information with interested information buyers. For example, one vehicle can buy semantic information \cite{vancea2018semantic}, \cite{liao2021road} from nearby vehicles/smart sensors about the conditions of the surrounding environment.

To promote the above two types of trade in the semantic communication system, i.e., semantic model trading and semantic information trading, we should design novel and effective incentive mechanisms:
\begin{enumerate}
    \item {\it{Semantic model trading:}} To encourage the participation of semantic model providers, incentive mechanisms are designed so that they are rewarded for supplying quality semantic models. In general, edge devices are willing to pay more for semantic models that can achieve better semantic performance. We are the first to propose a deep learning (DL) based auction mechanism to determine an allocation of the semantic model to the edge devices and the price to be paid by the edge devices to the model providers. We show analytically that the DL-based auction attains the properties of truthfulness while maximizing the revenue of the model providers.
    \item {\it{Semantic information trading:}} To facilitate the semantic information trading between multiple semantic information buyers and edge devices, e.g., vehicles that are interested in collecting semantic information about the conditions of the surrounding environment \cite{vancea2018semantic}, \cite{liao2021road}, we introduce a double auction mechanism to model the competition between the buyers and edge devices. In the auction, we propose semantic based valuation functions, i.e., the valuation of the information is a function of semantic performance of the edge devices. In particular, the semantic information buyers are willing to pay more for the semantic information with higher accuracy, and hence the edge devices have more incentive to obtain better models from the semantic model trading. Moreover, the proposed double auction mechanism shows the desired properties of individual rationality, incentive compatibility, and budget balance, which are all significant properties to achieve sustainable and rational trading.
\end{enumerate}
While many recent works have focused on improving the performance of semantic communication systems \cite{xie2020lite}, \cite{xie2021deep}, \cite{farsad2018deep}, few works have addressed the designs of incentive mechanisms for semantic communication systems. By achieving the aforementioned two kinds of trade, we propose a novel hierarchical trading system to enhance the economically-sustainable development of semantic communication systems. The main contributions of our paper are:
\begin{itemize}
    \item We propose an incentive design framework for the semantic model trading and semantic information trading to support the deployment of semantic communication systems. Our designed mechanisms support the development of semantic communication systems by motivating the participation of model providers to build and share high-quality semantic engines, buyers to acquire relevant and useful semantic information, and semantic information sellers to facilitate other stakeholders in the semantic information exchange.
     \item We model the competition in the semantic model trading and semantic information trading with auction mechanisms. Different from conventional auctions, our auction can maximize the revenue of semantic model providers while achieving the properties of individual rationality and incentive compatibility. Simulation results are provided based on a case study on semantic text transmission where we derive the valuation functions based on the sentence similarity score and bilingual evaluation understudy (BLEU) score \cite{papineni2002bleu}.
    
    \item We propose an effective feature reduction method for data transmission under a limited data transfer budget. In contrast to existing works of feature reduction techniques for semantic communication systems \cite{jiang2021deep}, \cite{yang2021semantic}, our method does not increase communication cost and reduce the performance gap between partial feature and full feature.

\end{itemize}
Compared with our previous work \cite{liew2022economics}, the significant extensions in this paper include:
\begin{itemize}
    \item In contrast to previous work in which the incentive mechanism is customized for wireless powered devices, we propose a general framework that can be applied to semantic communication systems with different purposes.
    \item While the previous work focuses on semantic information transfer, we consider both and joint semantic model trading and semantic information trading in this paper.
    \item To model realistic semantic communication systems, multiple semantic information buyers and sellers are considered instead of a single buyer setting in the previous work.
\end{itemize}

Our paper is organised as follows. In Section \ref{sec:relatedwork}, we discuss the related works of semantic communication systems and incentive mechanism design. In Section \ref{sec:sysmodel}, we detail the system model and problem formulation. In Section \ref{sec:casestudy}, we present a case study of semantic model trading and semantic information trading for semantic text transmission. In Section \ref{sec:results}, we present the numerical results, and Section \ref{sec:conclusion} concludes the paper.

\section{Related Work}\label{sec:relatedwork}
\subsection{Vehicular Networks}
With the development of vehicular infrastructure in recent years, vehicles can be seen as important network players with computing, caching and communication capabilities~\cite{li2021adaptive,ye2021joint}. However, as the number of vehicles increases, the vehicular network relies heavily on reliable real-time communication and interactions for complex operations~\cite{wu2020dynamic,nanda2019internet}, such as route planning and collision avoidance. Thus, timely and accurate information updates are vital to the development of the vehicular networks. This implies that the conventional communication paradigm which seeks the lowest latency is no longer a sustainable development direction. To make fast and accurate decisions in vehicular networks, it is important to leverage the semantic meaning of information~\cite{pappas2021goal}. The authors in~\cite{zhu2021video} design a resource allocation algorithm for semantic video transmission in vehicular networks. By using the proposed algorithm \cite{zhu2021video}, the semantic understanding accuracy of the video transmission is optimized by a multi-agent deep Q-network. The simulation results show that the proposed method can achieve as high as 70\% improvement for the density of correctly detected objects, compared with the conventional QoS and QoE based resource allocation methods. 

However, it is not realistic to train a usable semantic model for each vehicle, due to the limited computing resources and the dynamic positioning of vehicles~\cite{tayyaba20205g}. Therefore, we will consider a semantic model trading system in this paper. Moreover, considering the importance of semantic information in the vehicular networks, vehicles can then sell the semantic information to potential buyers. The trading of semantic information is gaining attraction especially for the sustainable development of large-scale multi-agent systems.

\subsection{Deep Learning Enabled Semantic Communication Systems}
Conventional communication systems focus on transmitting bits or symbols with minimum error from the transmitter to the receiver, and the performance is evaluated at the bits or symbols level. In contrast to the traditional communication systems, semantic communication system aims to communicate at the semantic level, where performance is evaluated by the recovery of the meanings of the data instead of bits accuracy. Semantic communication systems for text \cite{xie2021deep}, speech signals \cite{weng2021semantic}, and multimodal data \cite{xie2021task} first encode the data by a semantic encoder and send the encoded semantic information to the receivers. The receivers then decode the received signals with semantic decoders to recover the original data. Typically, the semantic encoders and decoders are implemented by end-to-end DL networks and trained with labeled data.

To improve the encoding efficiency, several works focus on reducing the size of the data during transmission. The authors in \cite{jiang2021deep} mask the bits according to the original sentence length to save the transmission resources. For the image classification task, the authors in \cite{yang2021semantic} use the gradient of the neural network to select import features. However, the proposed method requires extra storage cost to store the gradients of weights of the network. Most of the existing data reduction techniques are implemented together with the training process. A drawback is that, after the model is trained and parameters are fixed, further reduction of data size degrades the performance of the networks. To solve this problem, we develop an effective data reduction technique to reduce the performance gap in this paper.

\begin{figure*}[htb]
  \centering
  \centerline{\includegraphics[width=18cm]{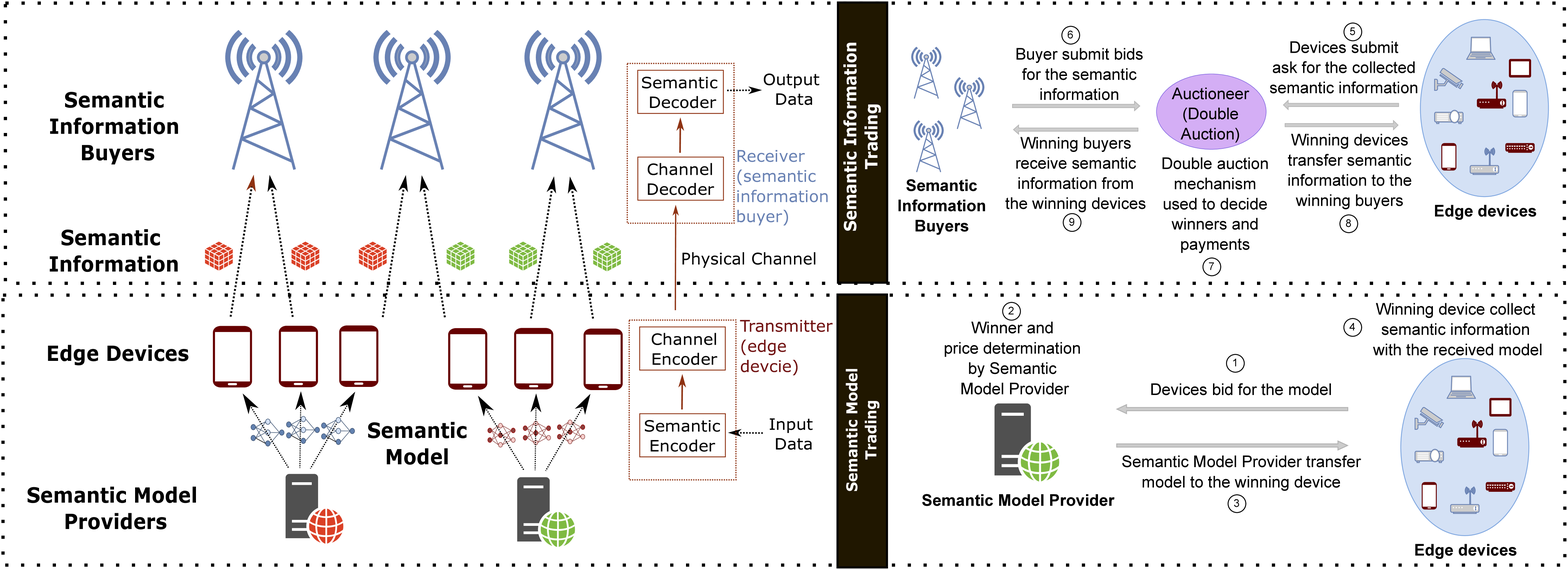}}
 \caption{The system model which includes semantic model trading and semantic information trading. In the semantic model trading, edge devices trade with semantic model providers to obtain semantic encoders/decoders for semantic communications. In the semantic information trading, edge devices equipped with semantic encoders/decoders trade semantic information with buyers.}
\label{fig:sysmod}
\end{figure*}

\subsection{Incentive Mechanism Design}
In real-world settings, data transmissions are limited by the communication resources such as bandwidth and energy. Incentive mechanisms are designed to encourage certain parties to contribute to a communication network. For example, in a multi-node wireless powered communication network, selfish wireless nodes are not willing to charge other nodes by consuming their resources. To encourage the participation of these nodes, \cite{zheng2021age} proposed incentive schemes to deal with the selfishness of wireless nodes with an Age of Information (AoI) based utility functions. In collaborative edge learning, incentive mechanisms are used to incentivise the data owners to provide the updated model parameters for global model aggregation \cite{lim2021decentralized}. 

Given that most of the communication networks are using conventional communication systems, semantic-aware incentive mechanisms are needed to be designed to motivate the participation of all parties in the development of semantic communication systems. We propose auctions as incentive mechanisms in the semantic model trading and semantic information trading, and derive the value of the semantic model and semantic information with semantic-based valuation functions.

\section{System Model and Problem Formulation}
\label{sec:sysmodel}
\subsection{System Model}
We consider a semantic communication network (Fig. \ref{fig:sysmod}) that consists of a set $\mathcal{M} = \{1,2,\ldots,m,\ldots,M\}$ of $M$ edge devices. To perform semantic encoding and decoding, the edge devices have to obtain the trained semantic models from the semantic model providers. Model trading is a common practice in collaborative edge learning, and in particular, federated learning \cite{lim2020federated}, where the model providers (sellers) receive incentives for providing trained models to the participants (buyers). In the case of semantic communications, the models being traded are the semantic encoders and decoders used for semantic information encoding and decoding, respectively. Devices with limited computation and communication resources can obtain high-quality semantic models from model trading. Moreover, it is shown that mismatches in communication channels and background knowledge of the communication environment degrade the performance of a pretrained semantic communication model \cite{xie2021deep}. Therefore, trading with model providers that perform machine learning training based on the relevant background knowledge and communication environment helps to improve the semantic performance of the devices. For example, devices can trade with the model providers that collect training data from the same certain geographical area as the buyers \cite{yang2022semantic}. 

To encourage the participation of model providers,  incentive mechanisms should be designed to ensure that model providers are appropriately rewarded from the semantic model trading process.
Similar to incentive mechanisms designed for the model trading in FL, the devices have to compete to obtain the semantic models from the semantic model providers. Intuitively, the devices are willing to pay more if the model obtained can achieve high semantic performance.

Besides, there exists a set $\mathcal{N} = \{1,2,\ldots,n,\ldots,N\}$ of $N$ semantic information buyers that are interested to obtain semantic information from the devices. For example, this may be semantic information trading between UAVs in real time \cite{yun2021attention}, and collection of semantic information for image classification tasks for autonomous vehicles \cite{yang2021semantic}. In this case, incentive mechanism design is also needed to facilitate the trading of such semantic information.

In the following, we propose two auction mechanisms for semantic model trading and semantic information trading. In the semantic model trading, we adopt a DL-based auction mechanism to derive the \textit{semantic-aware} valuation of the semantic models. The semantic model trading could be a channel to supply the semantic model for the devices to extract semantic information. Then, the semantic information from devices with higher accuracy is more valuable to semantic information buyers. For semantic information trading, we study the double auction mechanism for information trading between multiple buyers and multiple semantic information sellers. We further investigate how the semantic model obtained from the semantic model trading affects the results in the semantic information trading.

\subsection{Auction for Semantic Model Trading}
\label{smt}

The valuation of the devices for the model provided by the service provider is given by:
\begin{equation}\label{modalval}
    v_{m} = A_{p} - A_{m},
\end{equation}
where $A_{p}$ is the accuracy of the model from the model provider, and $A_{m}$ is the accuracy of the current model of device $m$ ($A_{m} = 0$ if the device does not own any model). The accuracy metric can be the text similarity score for semantic text transmission \cite{xie2021deep}, signal-to-distortion ratio (SDR) for semantic speech signal transmission \cite{weng2021semantic}, and answer accuracy in visual question answering (VQA) \cite{xie2021task}. In every round of the single-item auction, the model provider, i.e., the auctioneer collects bids $(\mathfrak{b}_1,\mathfrak{b}_2,\ldots,\mathfrak{b}_M)$ from all smarts devices, i.e., bidders, and then decides the winner, $m^*$, and corresponding payment price, $\theta_{m^*}$. The utility of the device is given by $u_m = v_m - \theta_{m^*}$, if the device is the winner and $u_m = 0$ otherwise. 

Traditional single-item auctions such as the first-price auction and Second-Price Auction (SPA) can be used to determine the winner and price. For an auction to be optimal \cite{myerson1981optimal}, it should attain the properties of incentive compatibility and individual rationality. Individual rationality guarantees that the utility of the devices is non-negative by participating in the auction, i.e., $u_m \geq 0$. Incentive compatibility ensures that each device submits bids according to their true valuations, respectively, i.e., $\mathfrak{b}_m = v_m$, regardless of the actions of other devices, and the utility of each device is maximized by submitting the truthful bid. In the first-price auction, the highest bidder wins and pays the exact bid submitted, maximizing the revenue gain of the model provider but does not guarantee incentive compatibility. In SPA, the highest bidder wins but pays the price of the second highest bidder. SPA ensures incentive compatibility but does not maximize the revenue of the model provider. 

\begin{algorithm}
  \caption{DL-Based Auction (DLA) Algorithm}\label{alg:aucalg}
  \label{alg:DLA}
  \hspace*{\algorithmicindent} \textbf{Input:} Bids of devices $\boldsymbol{\mathfrak{b}}=(\mathfrak{b}_1,\ldots,\mathfrak{b}_m,\ldots, \mathfrak{b}_M)$\\
  \hspace*{\algorithmicindent} \textbf{Output:} Winner and Payment Price
  
  \begin{algorithmic}[1]
    \STATE \textbf{Initialization:} $\mathbf{w} = [w_{qs}^m]\in \mathbb{R}_{+}^{M\times QS}, \bm{\beta} = [\beta_{qs}^m] \in \mathbb{R}^{M\times QS} $
    \WHILE{Loss function $\hat{R}(\mathbf{w,\bm{\beta}})$ is not minimized}
    \STATE Compute transformed bids $\overline{\mathfrak{b}}_m = \Phi_m(\mathfrak{b}_m) = \min_{q\in Q}\max_{s\in S}(w_{qs}^m\mathfrak{b}_m + \beta_{qs}^m)$ 
    \STATE Compute the allocation probabilities $z_m(\overline{\boldsymbol{\mathfrak{b}}}) = softmax(\overline{\mathfrak{b}}_1, \overline{\mathfrak{b}}_2,\ldots,\overline{ \mathfrak{b}}_{M+1};\kappa)$
    \STATE Compute the SPA-0 payments $\theta^0_m(\overline{\boldsymbol{\mathfrak{b}}}) = ReLU(\max_{j\neq m} \overline{\mathfrak{b}}_j)$
    \STATE Compute the conditional payment $\theta_m = \Phi_m^{-1}(\theta^0_m(\overline{\boldsymbol{\mathfrak{b}}}))$
    \STATE Compute the loss $\hat{R}(\mathbf{w} ,\bm{\beta})$
    \STATE Update parameters $\mathbf{w}$ and $\bm{\beta}$ using SGD optimizer

    \ENDWHILE
    \RETURN Winner $m^*$ and payment price $\theta_{m^*}$
  \end{algorithmic}
\end{algorithm}

We adopt a DL-based optimal auction mechanism \cite{dutting2019optimal} that can maximize the revenue of the seller while achieving the properties of incentive compatibility and individual rationality. The auctioneer (i.e., the model provider) does not have a priori knowledge about the bidders and optimal decisions in determining the winner. Nevertheless, the model provider can learn from experience and adjust the auction decision using DL-based optimal auction. 
The DL-based auction consists of three major functions: monotone increasing function, $\Phi_m$, allocation rule, $z_m$, and conditional payment rule, $\theta_m$. 
Firstly, the input bids, $\boldsymbol{\mathfrak{b}}=(\mathfrak{b}_1,\ldots,\mathfrak{b}_m,\ldots, \mathfrak{b}_M)$, are transformed by $Q$ groups of $S$ linear functions, followed by the $\min$ and $\max$ operations, i.e., the transformed bid, 
\begin{equation}
    \overline{\mathfrak{b}}_m= \Phi_m(\mathfrak{b}_m) = \min_{q\in Q}\max_{s\in S} (w_{qs}^m \mathfrak{b}_m + \beta_{qs}^m),
\end{equation}
where $w_{qs}^m\in \mathbb{R}_+$, $q=1,\dots,Q$, $s=1,\dots,S$ and $\beta_{qs}^m\in\mathbb{R}$, $q=1,\dots,Q$, $s=1,\dots,S$ are the weight and bias to be trained, respectively. 
The linear functions are strictly monotonically increasing functions to ensure the properties of incentive compatibility and individual rationality of the auction:
\newtheorem{theorem}{Theorem}
\begin{theorem}
(\cite{dutting2019optimal}) For any set of strictly monotonically increasing function \{$\Phi_1,\dots,\Phi_M$\}, an auction defined by allocation rule $z_m=z_m^0 \circ \Phi_m$ and the payment rule $\theta_m = \Phi_m^{-1} \circ \theta_m^0 \circ \Phi_m$ has the properties of incentive compatibility and individual rationality, where $z^0$ and $\theta^0$ are the allocation and payment rule of a second price auction with zero reserve, respectively, and $\circ$ indicates function composition, i.e., $(f\circ g)(x) = f(g(x))$.
\end{theorem}
To ensure that the auction learnt by the network achieves incentive compatibility and individual rationality, we constrain the allocation and payment rules of the network by following Theorem 1.
After the monotone transformation, the transformed bids are passed to separate networks that approximate the allocation and payment rule. 
The allocation rule which follows the second price auction with zero reserve (SPA-0) allocation rule is approximated by a softmax function \cite{bridle1989training} to maximize the allocation probability of the highest bid, i.e.,
\begin{equation}
    z_m(\overline{\boldsymbol{\mathfrak{b}}})=\frac{e^{\kappa \overline{\mathfrak{b}}_m}}{\sum^{M+1}_{j=1} e^{\kappa \overline{\mathfrak{b}}_j}},
\end{equation}
where $\overline{\mathbf{\mathfrak{b}}}=(\overline{\mathfrak{b}}_1,\dots,\overline{\mathfrak{b}}_{M+1})$, $\overline{\mathfrak{b}}_{M+1}$ is an additional dummy input, and $\kappa$ determines the quality of the approximation. The higher the value of $\kappa$, the higher the accuracy of approximation but the allocation function is less smooth and harder to optimize.
The SPA-0 payment rule is given by:
\begin{equation}
    \theta^0_m(\overline{\boldsymbol{\mathfrak{b}}}) = ReLU(\max_{j\neq m} \overline{\mathfrak{b}}_j),
\end{equation}
where $ReLU(x)=\max(x,0)$ is used to ensure that the payment is non-negative.
To obtain the payment price, the inverse transformation function is applied on the SPA-0 price of the transformed bids, i.e., 
\begin{equation}
    \theta_m = \Phi_m^{-1}(\theta^0_m(\overline{\boldsymbol{\mathfrak{b}}})),
\end{equation}
where the inverse transformation function can be expressed by:
\begin{equation}
    \Phi_m^{-1}(y)=\max_{q\in Q}\min_{s\in S}(w^m_{qs})^{-1}(y-\beta_{qs}^m).
\end{equation}

To maximize the revenue, the network optimizes a loss function that is the negative value of the seller revenue. The loss function is given by 
\begin{equation}
    \hat{R}(\mathbf{w} ,\bm{\beta})=-\sum^M_{m=1}z_m(\overline{\boldsymbol{\mathfrak{b}}})\theta_m.
\end{equation}

\subsection{Auction for Semantic Information Trading}
We consider $N$ semantic information buyers and $M$ devices, where the buyers are interested in buying semantic information from the devices. Consider that the devices obtain the semantic models from the semantic model trading, the semantic information buyers are willing to pay more for the semantic information from devices with high accuracy, $A_m$. 
We propose a single-round double auction for the one-to-one mapping of the buyers and the sellers. In the double auction, there are
\sloppy
\begin{itemize}
        \item A set of semantic information buyers $\mathcal{N} = \{1,\ldots,n,\ldots,N\}$ 
        \item A set of semantic information sellers $\mathcal{M} = \{1,\ldots,m,\ldots,M\}$, the devices that provide semantic information to the buyers
        \item A trusted third party, the auctioneer
\end{itemize}

Based on the semantic performance, each buyer has different preferences for the devices. Let $\mathbf{b_n} = (b_n^1,\ldots,b^M_n)$ denote the bid vector of buyer $n$, where $b_n^m$ is the bid of buyer $n$ for device $m$, i.e., the price that buyer $n$ is willing to pay for receiving semantic information from device $m$.

Let $\mathbf{a} = (a_1,\ldots,a_M)$ denote the ask vector of the devices, where $a_m$ is the ask of device $m$, i.e., the price that device $m$ is willing to receive for trading the semantic information. The value of the semantic information from device $m$ to buyer $n$ can be expressed as
\begin{equation}
    v_n^m = A_n^m(A_m),
\end{equation}
where $A_n^m$ is the accuracy of the semantic information transmitted by device $m$ to buyer $n$, and $A_m$ (determined by the semantic model trading) is the upper bound of the achievable accuracy of current semantic model. 

Let $p_n$ be the price that buyer $n$ pays, the utility of buyer $n$ is given by 
   \begin{equation}
   \label{iotutil}
       u_n^b =
           \begin{cases}
           v_n^m - p_n & \text{if buyer $n$ wins the auction,}\\
           0 & \text{otherwise.}
    \end{cases}     
   \end{equation}
Note that to compare the utility of buyer $n$ when it wins different devices, we also use $u_{n,m}^b$ and $u_{n,m'}^b$ to denote the utility of buyer $n$ when it wins to obtain semantic information of device $m$ and $m'$, respectively.

Following \cite{jiao2020toward}, the data collection cost is given by 
\begin{equation}
    c_m^{d} = d_m \gamma_m,
\label{dcost}
\end{equation}
where $d_m$ and $\gamma_m$ are the data size and unit data cost, respectively.
The computational cost can be formulated as
\begin{equation}
    c_m^{cp} = d_m \Gamma_m,
\label{cpcost}
\end{equation}
where $\Gamma$ is the unit computational cost to extract semantic information from the collected data. This cost can be due to the energy consumption \cite{lim2021decentralized} or edge/cloud computation resource rental fee \cite{dinh2020online}. The communication cost for device $m$ to transmit the semantic information is
\begin{equation}
    c_m^{cm} = P_m\frac{\mathbb{N}_m}{R}\nu_m,
\label{cmcost}
\end{equation}
where $P_m$ is the communication power, $\mathbb{N}_m$ is the number of bits used to represent the semantic information, $R$ is the transmission rate in bits per second, and $\nu_m$ is the unit energy cost for communication.
The cost of the semantic model is given by
\begin{equation}
    c_m^{md} = \frac{\theta_m}{T_m},
\label{mdcost}
\end{equation}
where $\theta_m$ is the price paid for the current semantic model (determined by the model trading auction in Section \ref{smt}) and $T_m$ is the expected number of transmissions with the model.

The total cost for device $m$ to transmit the semantic information is then defined as follows:
\begin{equation}
\begin{aligned}
    C_m &= c_m^{d} + c_m^{cp} + c_m^{cm} + c_m^{md}  \\ 
    & = d_m \gamma_m + d_m \Gamma_m + P_m\frac{\mathbb{N}_m}{R}\nu_m + \frac{\theta_m}{T_m}.
\label{cost}
\end{aligned}
\end{equation}

Let $y_m$ be the payment to device $m$, the utility of the device $m$ is given by
   \begin{equation}
       u_m^s =
           \begin{cases}
           y_m - C_m & \text{if device $m$ wins the auction,}\\
           0 & \text{otherwise.}
    \end{cases}
   \end{equation}

The proposed double auction has two stages, the {\it candidate-determination and pricing} stage, and the {\it candidate-elimination} stage. The algorithms for the two stages are shown in Algorithms \ref{alg:cdp} and \ref{alg:ce} respectively. Note that the DLA refers to the DL-Based Auction in Algorithm \ref{alg:DLA}. In the {\it candidate-determination and pricing} stage, the auctioneer determines the buyer candidates of each device, the prices that the buyer candidates pay, and the payment to be rewarded to the devices.

Let $n$ and $\theta_n$ denote the winning buyer and payment price determined by DLA, respectively. For each device $m$, all bids are sent to DLA to determine the winner and payment price. If the payment price is not lower than the ask $a_m$, i.e., $\theta_n \ge a_m$, then the buyer $n$ is added to the set of buyer candidates $\mathcal{N}_c$ with price $p_n = \theta_n$, and device $m$ is added to the set of seller candidates $\mathcal{M}_c$ with payment $y_m = \theta_n$.

\begin{algorithm}
  \caption{Candidate Determination and Pricing}\label{alg:cdp}
  \hspace*{\algorithmicindent} \textbf{Input:} $\mathcal{N},\mathcal{M},\mathbf{b},\mathbf{a}$ \\
  \hspace*{\algorithmicindent} \textbf{Output:} $\mathcal{N}_c,\mathcal{M}_c,\mathcal{P}_c,\mathcal{Y}_c$, $\hat{\sigma}$
  \begin{algorithmic}[1]
    \FOR{$ m \in \mathcal{M}$}
    \STATE $n, \theta_n= DLA(\{b^m_n,\forall n \in \mathcal{N}\})$
    \IF{$\theta_n \geq a_m$}
      \STATE $\hat{\sigma}(m) = n$, buyer $n$ is added into $\mathcal{N}_c$, and seller $m$ is added into $\mathcal{M}_c$
      \STATE $p_n = y_m = \theta_n$
      \STATE $p_n$ and $y_m$ are added into $\mathcal{P}_c$ and $\mathcal{Y}_c$, respectively
    \ENDIF
    

    \ENDFOR
    
  \end{algorithmic}
\end{algorithm}

After the first stage, each buyer candidate may win more than one device. In the {\it candidate-elimination} stage, for each buyer candidate, the algorithm selects the best device such that the buyer yields the highest utility in Equation \eqref{iotutil}. If more than one device yields the same highest utility for the buyer, the best device is randomly selected.

In the following, we prove that the double auction mechanism in our model satisfies the properties of individual rationality, incentive compatibility, and budget balanced.

\begin{algorithm}
  \caption{Candidate-elimination}\label{alg:ce}
  \hspace*{\algorithmicindent} \textbf{Input:} $\mathcal{N}_c, \mathcal{M}_c ,\mathcal{P}_c, \mathcal{Y}_c, \hat{\sigma}, \mathbf{b}$ \\
  \hspace*{\algorithmicindent} \textbf{Output:} $\mathcal{N}_w,\mathcal{M}_w,\mathcal{P}_w,\mathcal{Y}_w$
  \begin{algorithmic}[1]
  \STATE $\mathcal{N}_w \leftarrow \mathcal{N}_c,\mathcal{M}_w \leftarrow \mathcal{M}_c, \mathcal{P}_w \leftarrow \mathcal{P}_c,\mathcal{Y}_w \leftarrow \mathcal{Y}_c, \sigma \leftarrow \hat{\sigma} $ 
    \FOR{any two seller $m, m' \in \mathcal{M}_w, m \ne m'$}
      \IF{$\sigma(m) = \sigma(m')$}
        \IF{$u_{n,m}^b = u_{n,m'}^b$}
        \STATE $m^*\leftarrow$ randomly selected from $\{m,m'\}$
        \ELSE 
        \STATE $m^*\leftarrow$ arg $\min_{k\in\{m,m'\}}u_{n,k}^b$
    \ENDIF
    \ENDIF
    \STATE $\mathcal{M}_w \leftarrow \mathcal{M}_w \setminus \{m^*\}$
    \STATE $\mathcal{P}_w \leftarrow \mathcal{P}_w \setminus \{p_{m^*}\}$
    \STATE $\mathcal{Y}_w \leftarrow \mathcal{Y}_w \setminus \{y_{m^*}\}$
        
    \ENDFOR

    
  \end{algorithmic}
\end{algorithm}

\newtheorem{lemma}{Lemma}
\begin{theorem}
The proposed double auction mechanism is individually rational. All winning buyers and sellers are rewarded with non-negative utilities i.e. $p_n \le b_n^m$ and $y_m \ge a_m$ 
\end{theorem}
\begin{proof}
From Algorithm \ref{alg:cdp}, since DLA has the property of individual rationality \cite{dutting2019optimal}, we have $\theta_n \leq b_n^m$. Therefore $p_n \le b_n^m$ and $y_m \ge a_m$, individual rationality is satisfied in the candidate determination and pricing stage. Since Algorithm \ref{alg:ce} does not change the value of $p_n$ and $y_m$, the individual rationality is preserved after the candidate eliminations.


\end{proof}

\begin{theorem}
The proposed double auction mechanism is incentive compatible. All buyers and sellers submit their bids and asks truthfully as they cannot improve their utilities by submitting bids and asks that are different from their true valuations.
\end{theorem}

\begin{proof} 
We prove the incentive compatibility by the following lemmas:
\begin{enumerate}
    \item The proposed double auction mechansim is truthful for the sellers (as shown in Lemma \ref{lemma:selleric}).
    \item The proposed double auction mechanism is truthful for the buyers (as shown in Lemma \ref{lemma:buyeric}).
\end{enumerate}
\end{proof}

\begin{lemma}
The proposed double auction mechanism is truthful for the sellers.
\label{lemma:selleric}
\end{lemma}

\begin{proof} To prove that the proposed double auction mechanism is truthful for the sellers, we discuss the three possible outcomes for the sellers in the following subsets:
\begin{enumerate}
    \item Subset $\mathcal{M}_w$, sellers that win the auction,
    \item Subset $\mathcal{M}_c \setminus \mathcal{M}_w$, sellers that are selected as candidates but are eliminated during the candidate elimination stage, and
    \item Subset $\mathcal{M} \setminus \mathcal{M}_c$, sellers that are not selected as candidates.
\end{enumerate}
In each of the subsets, we discuss the cases where the sellers bid untruthfully. In each case, we show that the sellers cannot achieve higher utilities with the untruthful bids. Note that tilde $\widetilde{\cdot}$ is shown for the notations to indicate the outcomes of the untruthful cases.
\begin{enumerate}[leftmargin=*]
\item For seller $m \in \mathcal{M}_w$:

    {\it Case 1.} Seller $m$ does not win the auction with untruthful ask, $\widetilde{u}_m^s = 0 \le u_m^s$.
    
    {\it Case 2.} Seller $m$ wins the auction with untruthful ask. In this case, the payment does not change because the input bids to DLA are not changed, i.e., $\widetilde{u}_m^s = u_m^s$.
    
    

\item For seller $m \in \mathcal{M}_c \setminus \mathcal{M}_w$, changing ask does not change the price as discussed in the case of seller $m \in \mathcal{M}_w$. Therefore seller $m \in \mathcal{M}_c \setminus \mathcal{M}_w$ does not win the auction regardless of the value of $a_m$, $\widetilde{u}_m^s = u_m^s = 0$. 
\item For seller $m \in \mathcal{M} \setminus \mathcal{M}_c$:

    {\it Case 1.} Seller $m$ does not win by asking untruthfully, i.e., $m \notin \widetilde{\mathcal{M}}_w$, therefore the utility remains unchanged, $\widetilde{u}_m^s = u_m^s = 0$.
    
    {\it Case 2.} Seller $m$ wins by asking untruthfully, i.e., $m \in \widetilde{\mathcal{M}}_w$. Let buyer $n$ be the winner of semantic information from $m$ with price $\widetilde{p}_n = \widetilde{y}_m$. 
    To win the auction, $m$ has to ask lower than the true valuation such that $\widetilde{a}_m < C_m$. As the payment is not affected by $\widetilde{a}_m$, we have $\widetilde{y}_m = y_m$ and since $m$ does not win by asking truthfully, $y_m < C_m$, therefore $m$ suffers negative utility in this case, i.e., $\widetilde{u}_m^s = \widetilde{y}_m - C_m < 0 = u_m^s$.
    

Therefore we can conclude that the sellers cannot obtain a higher utility by asking untruthfully.

\end{enumerate}
\end{proof}
\begin{lemma}
The proposed double auction mechanism is truthful for the buyers.
\label{lemma:buyeric}
\end{lemma}
\begin{proof}
To prove that the proposed double auction mechanism is truthful for the buyers, we discuss the two possible outcomes for the buyers in the following subsets:
\begin{enumerate}
    \item Subset $\mathcal{N}_w$, buyers that win the auction, and
    \item Subset $\mathcal{N} \setminus \mathcal{N}_w$, buyers that lose the auction.
\end{enumerate}
In each of the subsets, we discuss the cases where the buyers ask untruthfully. In each case, we show that the buyers cannot achieve higher utilities with the untruthful asks. Note that tilde $\widetilde{\cdot}$ is shown for the notations to indicate the outcomes of the untruthful cases.
\begin{enumerate}[leftmargin=*]
    \item For buyer $n \in \mathcal{N}_w$, assuming $n$ wins seller $m$ by bidding truthfully. Let us consider the following cases when buyer $n$ bids untruthfully:
    
    {\it Case 1.} Buyer $n$ loses with untruthful bid, $\widetilde{u}_n^b = 0 \le u_n^b$.
    
    {\it Case 2.} Buyer $n$ wins the same seller $m$ with untruthful bid, given individual rationality property of DLA, we have $\widetilde{u}_n^b \le u_n^b$.
    
    {\it Case 3.} Buyer $n$ wins with a different seller $m'$ with untruthful bid. 
    Let us consider the following cases when buyer $n$ bids truthfully:
    \begin{itemize}[leftmargin=*]
        \item Seller $m' \in \mathcal{M}_c$ and $\hat{\sigma}(m')=n$. 
        Since buyer $n$ wins $m$ in the truthful case, we have $u_{n,m'}^b \le u_{n,m}^b$. Given that DLA has the property of individual rationality, we have $\widetilde{u}_{n,m'}^b \le u_{n,m'}^b$.
        Thus we know that $\widetilde{u}_{n,m'}^b \le u_{n,m}^b$. 
        \item Seller $m' \in \mathcal{M}_c$ and $\hat{\sigma}(m') \ne n$. It means that there is another buyer candidate $n'$ with higher or equal bid for $m'$, i.e., $b_{n'}^{m'} \ge b_n^{m'}$. When buyer $n$ wins $m'$ by bidding untruthfully, since DLA satisfies the individual rationality constraint, we have $\widetilde{u}_{n,m'}^b \le 0$. From Theorem 2, we know that $u_{n,m}^b \ge 0$ (all winning buyers and sellers are rewarded with non-negative utility), thus we have $\widetilde{u}_{n,m'}^b \le u_{n,m}^b$.
        \item Seller $m' \notin \mathcal{M}_c$ and buyer $n$ wins $m'$ by bidding untruthfully. Since DLA has the property of individual rationality, we have $\widetilde{u}_{n,m'}^b \le 0$. From Theorem 2, we know that $u_{n,m}^b \ge 0$, thus we have $\widetilde{u}_{n,m'}^b \le u_{n,m}^b$.
    \end{itemize}

    \item For buyer $n \in \mathcal{N}\setminus\mathcal{N}_w$ with utility $u_n^b = 0$. We consider the following cases when buyer $n$ bids untruthfully. 
    
    {\it Case 1.} Buyer $n$ loses with untruthful bid, $\widetilde{u}_n^b = 0 = u_n^b$.

    {\it Case 2.} Buyer $n$ wins seller $m$ by bidding untruthfully. Since DLA has the property of individual rationality, we have $\widetilde{u}_n^b \le 0 = u_n^b$.

\end{enumerate}
Therefore we can conclude that the buyers cannot obtain a higher utility by bidding untruthfully.
\end{proof}

\begin{theorem}
The proposed double auction mechanism is budget balanced. The total price paid by the winning buyers is not less than the total payment to the winning sellers, i.e., $\sum_{n \in \mathcal{N}_w} p_n \ge \sum_{m \in \mathcal{M}_w} y_m$.
\end{theorem}

\begin{proof}
According to Algorithms \ref{alg:cdp} and \ref{alg:ce}, the price that winning buyers pay and the payment received by winning sellers are equal for every winning seller-buyer pairs. Thus, we have
\begin{equation}
    \sum_{n \in \mathcal{N}_w} p_n - \sum_{m \in \mathcal{M}_w} y_m = 0.
\end{equation}
We can conclude that the double auction mechanism is budget balanced.
\end{proof}

In Algorithm \ref{alg:cdp}, since there are $M$ sellers in set $\mathcal{M}$, the time complexity of the candidate determination and pricing stage is $O(M)$. In Algorithm \ref{alg:ce}, we know that $|\mathcal{M}_c| \le |\mathcal{M}|=M$. In the worst case, the for-loop runs for $\frac{M(M-1)}{2}$ times. Therefore, Algorithm \ref{alg:ce} has the time complexity of $O(\frac{M(M-1)}{2})=O(M^2)$. Overall, the proposed double auction mechanism is a polynomial time algorithm with the time complexity of $O(M^2)$.

\section{Case Study: Semantic Text Transmission}\label{sec:casestudy}
In this section, we apply the proposed auction mechanisms to the semantic text transmission. We derive the valuations of the semantic model trading and semantic information trading for semantic text transmission.

\subsection{Deep Learning Enabled Semantic Communication Systems}
\label{dlsemcom}
We consider the $M$ devices perform text data transmission with DL enabled semantic communication systems, e.g., voice controlled devices (Google Nest Hub \footnote{https://www.cnet.com/home/smart-home/how-to-set-up-your-new-google-nest-hub-or-nest-hub-max/}, Amazon Echo \footnote{https://www.androidauthority.com/amazon-echo-5th-gen-3095027}, and Apple HomePod \footnote{https://www.apple.com/sg/newsroom/2021/10/apple-introduces-homepod-mini-in-new-bold-and-expressive-colors/}). In DL enabled semantic communication system, collected sentences, $\mathbf{S}=[s_1,s_2,\ldots,s_{N_s}]$, are encoded by semantic encoder and channel encoder. The encoded signal can be represented by
\begin{equation}
    \mathbf{X} = enc_c(enc_s(\mathbf{S})) \label{enc}, 
\end{equation}
where $\mathbf{X} \in \mathbb{R}^{N_s \times L \times D}$, $N_s$ is the number of sentences, $L$ is the sentence length, $D$ is the output dimension of channel encoder, $enc_c(\cdot)$ is the channel encoder, and $enc_s(\cdot)$ is the semantic encoder. Note that all inputs are padded to length $L$ before passing to the encoders. 
After winner determination of the double auction, winning devices transmit encoded information to the winning buyers. At the buyer, signal received can be expressed as 
\begin{equation}
    \mathbf{Y} = \mathbf{HX}+\mathbf{A},
\end{equation}
where $\mathbf{H}$ is the channel gain between the transmitter and receiver and $\mathbf{A} \sim \mathcal{N}(0,\sigma^2_n)$ is the additive white Gaussian noise (AWGN). The decoded sentences are given by 
\begin{equation}
    \widehat{\mathbf{S}} = dec_s(dec_c(\mathbf{Y})),
\end{equation}
where $dec_s(\cdot)$ and $dec_c(\cdot)$ are the semantic decoder and channel decoder of the receiver.

We adopt the network architecture of DeepSC \cite{xie2021deep} where the semantic encoder and decoder are implemented as multiple Transformer \cite{vaswani2017attention} encode and decode layers, and channel encoder as dense layers with different units. Our incentive mechanism can be easily extended to other network architectures by following the same evaluation procedure.

The BLEU score and the sentence similarity are two of the critical performance metrics of the text-based semantic communication system. The BLEU score measures an exact matching of words in the original and recovered sentences without considering their semantic information. In contrast to the BLEU score, the sentence similarity is calculated by the cosine similarity of the extracted semantic features from original and recovered sentences. In our model, a pre-trained Bidirectional Encoder Representations from Transformers (BERT) \cite{devlin-etal-2019-bert} model is used for the semantic features extraction. Let $\mathbf{s}$ and $\widehat{\mathbf{s}}$ denote one sentence from $\mathbf{S}$ and $\widehat{\mathbf{S}}$, respectively. The BLEU score can be expessed as 
\begin{equation}
    \log \text{BLEU} = \min \left(1-\frac{l_{\widehat{\mathbf{s}}}}{l_{\mathbf{s}}},0\right) + \displaystyle\sum_{i=1}^{I} u_i \log p_i,
\end{equation}
where $l_{\mathbf{s}}$ and $l_{\widehat{\mathbf{s}}}$ are the lengths of the original and recovered sentences respectively, $u_i$ is the weight of $i$-grams, and $p_i$ is the $i$-grams score, which is given by 
\begin{equation}
    p_i = \frac{\displaystyle\sum_{k=1}^{K_i} \min(C_k (\widehat{\mathbf{s}}), C_k(\mathbf{s}))}{\displaystyle\sum_{k=1}^{K_i} \min(C_k(\widehat{\mathbf{s}}))},
\end{equation}
where $K_i$ is the number of elements in $i$-th grams, and $C_k (\cdot)$ is the frequency count function for the $k$-th element in $i$-th grams.
The sentence similarity is given by 
\begin{equation}
    similarity(\widehat{\mathbf{s}},\mathbf{s}) = \frac{\mathbf{B}(\mathbf{s}) \cdot \mathbf{B}(\widehat{\mathbf{s}})^T}{\lVert\mathbf{B}(\mathbf{s})\rVert \lVert\mathbf{B}(\widehat{\mathbf{s}})\rVert},
\end{equation}
where $\mathbf{B}(\cdot)$ is a pre-trained BERT model used to measure the sentence similarity. 

In general, to obtain a higher BLEU score and similarity score, we need to increase the output dimension $D$ of the encoder \cite{devlin-etal-2019-bert}. 
However, increasing $D$ comes at the cost of a larger data size, and the amount of data that devices can send is limited by the communication resources, e.g., energy supply to the devices \cite{liew2022economics}. Specifically,  the BLEU score and the similarity score of device $m$ can be expressed as 
\begin{equation}
\label{simscore}
    s_m = f_{sim}(D) = f_{sim}(\frac{\mathbb{N}_m}{N_s\times L \times b_f}),
\end{equation}
and
\begin{equation}
\label{bleuscore}
    BLEU_m = f_{BLEU}(D) = f_{BLEU}(\frac{\mathbb{N}_m}{N_s\times L \times b_f}),
\end{equation}
respectively, where  $f_{sim}(\cdot)$ and $f_{BLEU}(\cdot)$ are simple lookup to obtain the scores of the model, $b_f$ is the number of bits used by a unit feature, and $\mathbb{N}_m$ is the total number of bits that the device $m$ can transmit. The values of $f_{sim}(\cdot)$ and $f_{BLEU}(\cdot)$ can be obtained by using different output dimension $D$ to evaluate the similarity score and the BLEU score, respectively. A unit feature is a single entry of $\mathbf{X} \in \mathbb{R}^{N_s \times L \times D}$, and $b_f$ is the number of bits used to represent a float type data. In our model, the data size in Equations \eqref{dcost} and \eqref{cpcost} is given by the number of words collected, i.e., $d=N_s\times L$. The total number of bits affects the communication cost as shown in Equation \eqref{cmcost}.

From Equations \eqref{simscore} and \eqref{bleuscore}, it is clear that the scores are affected by the size of data and model performance. Since each device has a different model performance and data to be sent, the similarity score and the BLEU score are different among the devices. 

\subsection{Semantic-Aware Valuation for Auctions}

In the semantic model trading, the devices bid according to the performance of the semantic model (Equation \eqref{modalval}), i.e., 
\begin{equation}
    \mathfrak{b}_m = v_m = A_{p} - A_m.
\end{equation}
The accuracy of the model from the model provider can be expressed as follows:
\begin{equation}
    A_{p} = \lambda_m s_{p} + \beta_m BLEU_{p},
\end{equation}
where $s_{p}$ and $BLEU_{p}$ are the similarity score and the BLEU score achievable by the model provided, respectively, $\lambda_m$ is the preference for the similarity score by the device $m$, $\beta_m$ is the preference for the BLEU score by the device $m$, and $\lambda_m + \beta_m=1$. If $\beta_m > \lambda_m$, it indicates that the device has more interest in the exact recovery of words whereas $\lambda_m > \beta_m$ indicates higher interest in the matching of the semantic meaning. For example, some medical devices \cite{dhyani2021intelligent} would have higher $\beta_m$ because the exact recovery of medical terms is more important, whereas devices that collect data for text classification \cite{shah2016review} would have higher $\lambda_m$.

The accuracy of the current model of device $m$ is given by:
\begin{equation}
    A_{m} = \lambda_m s_{m} + \beta_m BLEU_{m},
\end{equation}
where $s_{m}$ and $BLEU_{m}$ are the similarity score and the BLEU score achievable by the current model.
In the semantic information trading, based on the communication environment and resources, each device can achieve different semantic performance when transmitting information to the buyers. Therefore, based on the semantic performance, each buyer has different preferences for the devices. 

The value of the semantic information from device $m$ to buyer $n$ is given by:
\begin{equation}
v_n^m = \lambda_n s_m + \beta_n BLEU_m ,
\end{equation}
where $\lambda_n$ is the preference for the similarity score, and $\beta_n$ is the preference for the BLEU score by the buyer $n$. As the auction is truthful for all buyers and sellers, the buyers and sellers submit bids and asks according to their true valuations, i.e., $b_n^m = v_n^m$ and $a_m = C_m$. Again, the cost of collecting the information by device $m$ can be obtained from Equation \eqref{cost}. 

\subsection{Feature Reduction Technique}
\label{featreduce}
Let $\mathbb{N}_m$ denote the number of bits that device $m$ can send to the buyer. 
Based on the bit budget $\mathbb{N}_m$, not all features of the encoded information can be sent. 
However, the semantic communication model is trained with a fixed number of features with output dimension $D$. A sample of feature representation output by semantic encoder with 16 features  is shown in Fig. \ref{fig:textfeat}.
Sentences decoded from partial features have a lower similarity score and BLEU score than that decoded from all features.
Deep neural networks need to fine-tune the model parameters to reduce the gap in performance. 
Unfortunately, devices that operate on limited resource might not be able to fine-tune the model in real-time because it is both time and energy consuming.
Therefore an effective feature reduction method is required for these devices to minimize the gap in performance when they have to communicate with a limited bit budget.

\begin{figure*}[htb]
  \centering
  \centerline{\includegraphics[width=18cm]{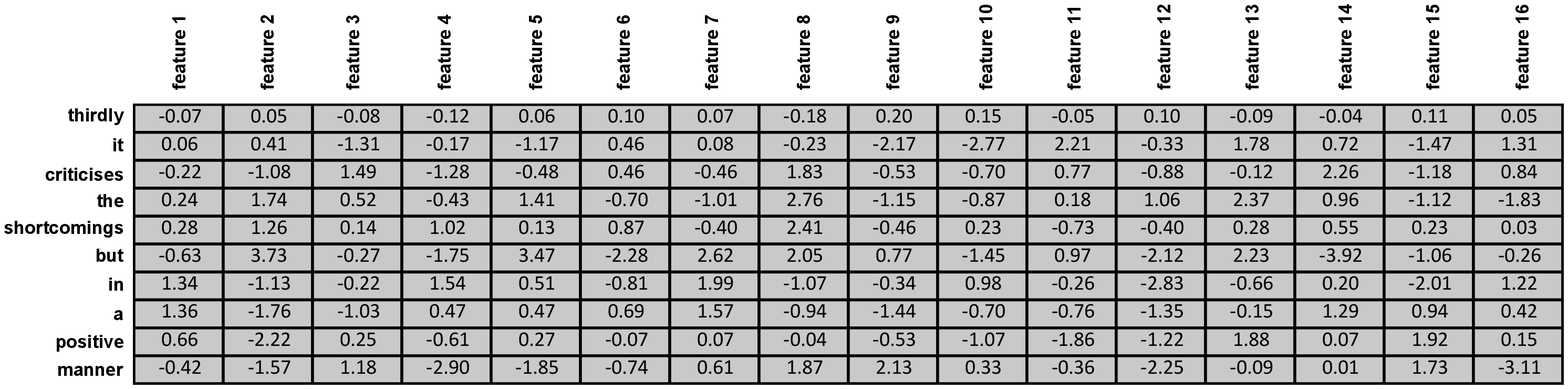}}
\caption{Sample output of semantic encoder with 16 features, input sentence: ``thirdly it criticises the shortcomings but in a positive manner".}
\label{fig:textfeat}
\end{figure*}

\begin{figure*}[htb]
  \centering
  \centerline{\includegraphics[width=16cm]{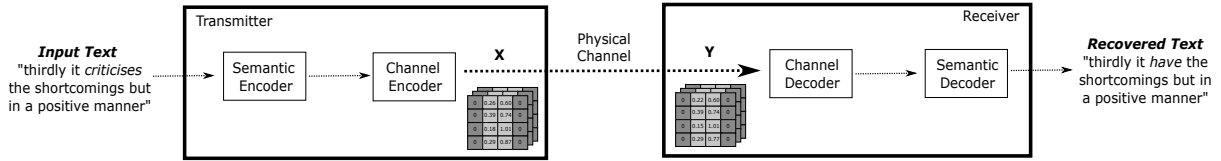}}
\caption{Illustration of controlled dropout during the training.}
\label{fig:cdroptrain}
\end{figure*}

\begin{figure}[t]

\begin{minipage}[b]{.48\linewidth}
  \centering
  \centering
  \centerline{\includegraphics[width=3.5cm]{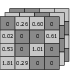}}
  
  \centerline{(a)}\medskip
\label{fig:convendrop}


\end{minipage}
\hfill
\begin{minipage}[b]{0.48\linewidth}
  \centering
\centering
  \centerline{\includegraphics[width=3.5cm]{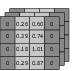}}
    \centerline{(b)}\medskip

\label{fig:rev}
\end{minipage}
\caption{An illustration of (a) conventional dropout and (b) controlled dropout.}
\label{fig:dropout}
\end{figure}

We propose a simple feature reduction method where the performance can be adjusted by a regularization technique \cite{ko2017controlled} during the training of the model. Consider that the model on device is pre-trained with output dimension $D$, under the limited bit budget, the encoded signal, $\mathbf{X} \in \mathbb{R}^{N_s \times L \times D}$ is reduced to $\mathbf{X}' \in \mathbb{R}^{N_s \times L \times D'}$, where $0<D'<D$. At the receiver, the received signal $\mathbf{Y}'\in \mathbb{R}^{N_s \times L \times D'}$ is padded with zeros to become $\mathbf{Y}\in \mathbb{R}^{N_s \times L \times D}$. The proposed data reduction method is illustrated in Fig. \ref{fig:cdroptrain}. To obtain $f_{sim}(\cdot)$ and $f_{BLEU}(\cdot)$, we first train the DeepSC model with the data with dimension $D$ and use the trained model to evaluate the similarity scores for output dimension $d$, $\forall d \in [1,D]$. Then, we can obtain $f_{sim}(\cdot)$ and $f_{BLEU}(\cdot)$ from the evaluation results of test datasets.

To reduce the degradation of performance, we add a controlled dropout \cite{ko2017controlled} layer before the channel decoding layer of the receiver. For example, if index $d_i$ is selected by controlled dropout, all units from (0, 0, $d_i$) to ($N_s-1$, $L-1$, $d_i$) become zeros. The conventional dropout \cite{srivastava2014dropout} technique randomly drops units (Fig. \ref{fig:dropout}(a)) in the training process to solve the overfitting issue of the deep neural networks. In contrast to conventional dropout, controlled dropout drops units intentionally, i.e., dropping a selected index of a dimension, as shown in Fig. \ref{fig:dropout}(b). An illustration of the effect of controlled dropout during training is shown in Fig. \ref{fig:cdroptrain}. In our experiments, we drop units from a certain index of the output dimension. As shown in \cite{ko2017controlled}, we can obtain a better performance than conventional dropout when the index is randomly selected. Following \cite{ko2017controlled}, the index is randomly selected with a dropout rate, $p_{drop}$, $0 < p_{drop} < 1$. Controlled dropout helps the model to generalize to the reduced features during the training.

\section{Numerical Results}\label{sec:results}
In this section, we evaluate the performance of the proposed auction mechanisms and feature reduction method. The values of experiment parameters are presented in Table \ref{tab:param}. The similarity and BLEU scores are sampled according to the simulation settings in \cite{lim2021decentralized}, \cite{luong2018optimal} for the DL-based auction. The dropout rate is set according to \cite{ko2017controlled}. Following \cite{jiao2020toward}, \cite{chen2022performance}, we set the cost-related parameters in the double auction as shown in Table \ref{tab:param}.
\begin{table}
\begin{center}
\caption{Experiment Parameters \cite{lim2021decentralized}, \cite{jiao2020toward}, \cite{ko2017controlled}, \cite{luong2018optimal},   \cite{chen2022performance}}

\begin{tabular}{ |l|l|c|c| } 

\hline
Parameters & Values \\
\hline

Similarity score coefficient, $\lambda_n$ & $\sim \mathcal{U}[0,1] $\\ 
BLEU score coefficient, $\beta_n$ & $1-\lambda_n$ \\
 
Dropout rate, $p_{drop}$ & 0.1 \\
Reduced output dimension, $D'$ & $\sim \mathcal{U}[1,16] $ \\
Data size, $d_m$ & $\sim \mathcal{U}[10,100]$ \\
Unit data cost, $\gamma_m$ & $0.001$ \\
Unit computational cost, $\Gamma_m$ & $0.001$ \\
Communication power, $P_m$ & $1$ \\
Number of bits transmitted, $\mathbf{N}_m$ & $10000$ \\
Transmission rate, $R$ & $100000$ \\
Unit energy cost, $\nu_m$ & $0.01$ \\
Expected number of transmissions, $T_m$ & $100$ \\ 
\hline
\end{tabular}

\label{tab:param}
\end{center}
\end{table}
\subsection{Evaluation of DeepSC with Feature Reduction}
We first investigate the improvement of semantic performance under the proposed feature reduction method. With the help of the DeepSC, we set the output dimension of encoder $D$ to 16, and train the model under AWGN channel for 200 epochs. The training and test data is obtained from the proceedings of the European Parliament \cite{koehn2005europarl}. We use $7347$ English sentences in the dataset for our evaluation, and use the rest of the English sentences for training. The performance scores are considered for the evaluation of the proposed double auction mechanism.

As described in Section \ref{featreduce}, we add a controlled dropout layer between the physical layer and receiver. The dropout probability is set as $0.1$ which means $10\%$ of the features are dropped randomly in a controlled setting. We record the similarity score and the BLEU score for the output dimensions from 1 to 16, which is shown in  Fig. \ref{fig:odvsscore}. Regardless of the application of controlled dropout, we observe that the performance degrades as the output dimension decreases. The reason is that fewer features are transmitted.
However, when the output dimension changes from $1$ to $15$, the performance of model trained with controlled dropout outperforms constantly the baseline mode. In other words, as the output dimension decreases, the baseline model has a larger performance gap compared to the model with controlled dropout. Specifically, the reduction of the similarity score per output dimension is 0.05 in the baseline model and 0.04 in the proposed model. For the reduction of the BLEU score per output dimension, it is 0.06 for the baseline model and 0.05 for the proposed model. This result shows that the proposed model can maintain a similarity score of 0.80 even after 25\% of feature reduction ($D=12$) whereas the baseline model can only achieve the similarity score of 0.60 with the same output dimension. As shown in Table \ref{tab:sent}, the recovered sentence has higher similarity when the controlled dropout is applied.

However, we notice that the best performance achieved by the baseline model at $D=16$ is slightly higher than that of the proposed model. The BLEU score and the similarity score for our proposed model are 0.89 and 0.91, respectively, but both scores are 0.94 for the baseline model. The reason is that the accuracy is slightly dropped due to the generalization of the feature reduction. Overall, the gap in performance at fewer output dimensions is compensated by the controlled dropout during training. 

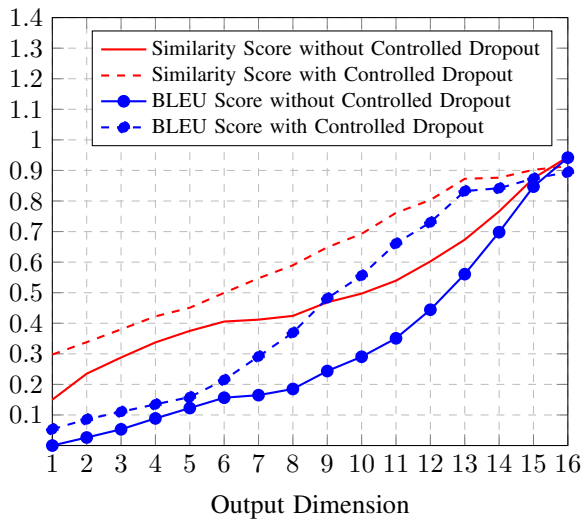
\begin{figure}[t]
\centering
\begin{tikzpicture}

\begin{axis}[
    xlabel={Output Dimension},
    xmin=1, xmax=16,
    ymin=0, ymax=1.4,
    xtick={1,2,3,4,5,6,7,8,9,10,11,12,13,14,15,16},
    ytick={0.1,0.2,0.3,0.4,0.5,0.6,0.7,0.8,0.9,1.0,1.1,1.2,1.3,1.4},
    legend pos=north east,
    legend style={fill=white, nodes={scale=0.8, transform shape}},
    ymajorgrids=true,
    xmajorgrids=true,
    grid style=dashed,
    legend cell align={left},
    every axis plot/.append style={thick},
    /pgf/number format/.cd,
    1000 sep={},
]
\pgfplotstableread{simscore_snr9.txt}\simscore;
\addplot[
    color=red,
    ]
    table
    [
    x expr=\thisrowno{0},
    y expr=\thisrowno{1}
    ] {\simscore};
    \addlegendentry{Similarity Score without Controlled Dropout};
    
\addplot[
    color=red,
    dashed,
    ]
    table
    [
    x expr=\thisrowno{0},
    y expr=\thisrowno{2}
    ] {\simscore};
    \addlegendentry{Similarity Score with Controlled Dropout};
    
\addplot[
    color=blue,
    mark=*,
    ]
    table
    [
    x expr=\thisrowno{0},
    y expr=\thisrowno{3}
    ] {\simscore};
    \addlegendentry{BLEU Score without Controlled Dropout};
    
\addplot[
    color=blue,
    mark=*,
    dashed,
    ]
    table
    [
    x expr=\thisrowno{0},
    y expr=\thisrowno{4}
    ] {\simscore};
    \addlegendentry{BLEU Score with Controlled Dropout};
    
\end{axis}
\end{tikzpicture}
\caption{BLEU and similarity score under different output dimensions.}
\label{fig:odvsscore}
\end{figure}

\begin{table*}
\begin{center}
\caption{Sample Sentences with and without Controlled Dropout}

\begin{tabular}{ |l|l|c|c| } 

\hline
\textbf{Original Sentence} & thirdly it criticises the shortcomings but in a positive manner \\
\hline
\textbf{Output Dimension = 15, with controlled dropout}  & thirdly it have the shortcomings but in a positive manner \\
\hline
\textbf{Output Dimension = 14, with controlled dropout} & thirdly it have the shortcomings but in a positive manner \\
\hline
\textbf{Output Dimension = 15, without controlled dropout} & thirdly it forward the shortcomings but in a positive manner\\
\hline
\textbf{Output Dimension = 14, without controlled dropout} & thirdly it played the shortcomings but in a off manner \\
\hline
\end{tabular}

\label{tab:sent}
\end{center}
\end{table*}

\subsection{Evaluation of Deep Learning based Auction Mechanism}
Without loss of generality, we consider that the devices do not own any semantic model initially, i.e., $A_{m} = 0$ and $\mathfrak{b}_m = v_m = A_{p}$. To obtain the bid profiles, we consider $A_{p} \sim U[0,0.4]$ and $\sim U[0.5,0.9]$. We collect 1000 training samples with 10 bidders (devices) in each of the samples and perform training for 500 epochs. From Fig. \ref{fig:revmodprovider}, we observe that the DL-based auction can always achieve higher revenue than that of the SPA, regardless of the values of $A_{p}$. The reason is that the DL-based auction mechanism can adapt to different bid profiles by optimizing the parameters in the DL network. Moreover, we observe that, while SPA is incentive compatible, it does not maximize the revenue of the model providers. In contrast to SPA, the DL-based auction maximizes the revenue of model providers while keeping the desired properties of incentive compatibility and individual rationality, which helps to attract more model providers to offer quality semantic encoder/decoder for semantic communications.

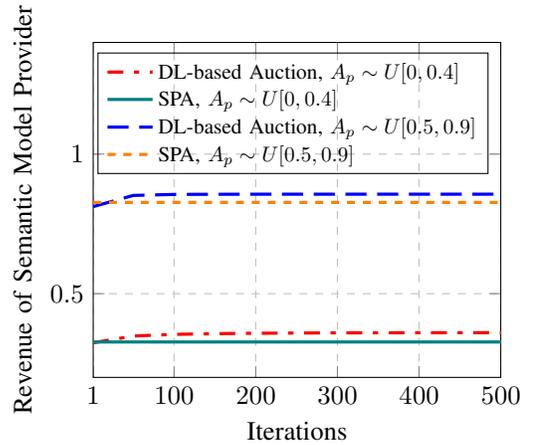
\begin{figure}
\centering
\begin{tikzpicture}
\pgfplotsset{
    width=7cm, compat=1.3
}
\begin{axis}[
    xlabel={Iterations},
    ylabel={Revenue of Semantic Model Provider},
    xmin=1, xmax=500,
    ymin=0.2, ymax=1.4,
    xtick={1,100,200,300,400,500},
    legend pos=north east,
    legend style={fill=none, nodes={scale=0.8, transform shape}},
    ymajorgrids=true,
    xmajorgrids=true,
    grid style=dashed,
    legend cell align={left},
    every axis plot/.append style={very thick},
    /pgf/number format/.cd,
    1000 sep={},
]
\pgfplotstableread{auction_res_uni_d50.txt}\resauction;
\addplot[
    color=red,
    dash pattern=on 6pt off 4pt on 2pt off 4pt,
    ]
    table
    [
    x expr=\thisrowno{0},
    y expr=\thisrowno{1}
    ] {\resauction};
    \addlegendentry{DL-based Auction, $A_{p}\sim U[0,0.4]$};
    
\addplot[
    color=teal,
    ]
    table
    [
    x expr=\thisrowno{0},
    y expr=\thisrowno{2}
    ] {\resauction};
    \addlegendentry{SPA, $A_{p}\sim U[0,0.4]$};
    
\addplot[
    color=blue,
    dashed,
    dash pattern= on 8pt off 4pt,
    ]
    table
    [
    x expr=\thisrowno{0},
    y expr=\thisrowno{3}
    ] {\resauction};
    \addlegendentry{DL-based Auction, $A_{p}\sim U[0.5,0.9]$};
    
\addplot[
    color=orange,
    dashed,
    ]
    table
    [
    x expr=\thisrowno{0},
    y expr=\thisrowno{4}
    ] {\resauction};
    \addlegendentry{SPA, $A_{p}\sim U[0.5,0.9]$};
    
\end{axis}
\end{tikzpicture}
\caption{Revenue of semantic model providers.}
\label{fig:revmodprovider}
\end{figure}

\subsection{Evaluation of Double Auction Mechanism}
To evaluate the performance of the double auction mechanism, we generate 1000 samples and average the simulation results. 
We set the number of sellers to $M=20$ and evaluate the performance under different number of buyers.
Note that in the following discussion, we refer semantic information buyers as buyers and devices as sellers for simplicity.

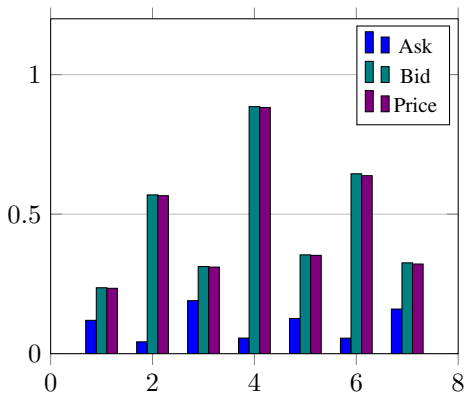
\begin{figure}
\centering
\begin{tikzpicture}
\pgfplotsset{
    width=7cm, compat=1.3
}
\begin{axis} [
    ybar=0pt,
    bar width=4pt,
    xmin=0,xmax=8,
    ymin=0,ymax=1.2,
    legend style={fill=white, nodes={scale=0.8, transform shape}},
    ymajorgrids=true]

\pgfplotstableread{n100_res.txt}\truthful;
\addplot[
    fill=blue,
    ]
    table
    [
    x expr=\thisrowno{0},
    y expr=\thisrowno{1}
    ] {\truthful};
    \addlegendentry{Ask};
    
\addplot[
    fill=teal,
    ]
    table
    [
    x expr=\thisrowno{0},
    y expr=\thisrowno{2}
    ] {\truthful};
    \addlegendentry{Bid};
    
\addplot[
    fill=violet,
    ]
    table
    [
    x expr=\thisrowno{0},
    y expr=\thisrowno{3}
    ] {\truthful};
    \addlegendentry{Price};

\end{axis}

\end{tikzpicture}

\caption{Ask, bid, and price of winning seller-buyer pairs.}
\label{bidaskprice}
\end{figure}


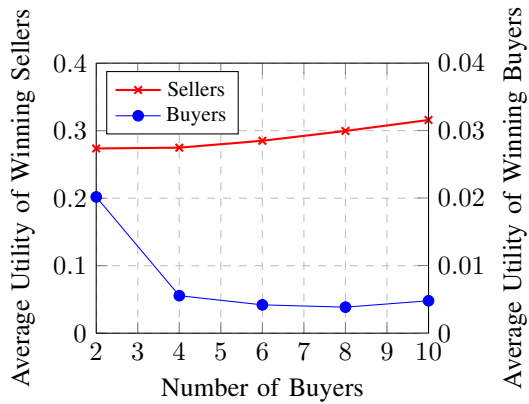
\begin{figure}
\centering
\pgfplotsset{scaled y ticks=false}
 \begin{tikzpicture}
\pgfplotsset{
    xmin=2, xmax=10,
    width=6cm, compat=1.3
}

\begin{axis}[
  axis y line*=left,
  ymin=0, ymax=0.4,
  xlabel=Number of Buyers,
  ylabel=Average Utility of Winning Sellers,
  xtick={2,3,4,5,6,7,8,9,10},
  ytick={0,0.1,0.2,0.3,0.4},
  ymajorgrids=true,
  xmajorgrids=true,
  grid style=dashed,
  legend cell align={left},
  every axis plot/.append style={thick},
  /pgf/number format/.cd,
  1000 sep={}
]
\addplot[mark=x,red]
  coordinates{
    (2,0.273609416)
    (4,0.274812524)
    (6,0.284980794)
    (8,0.29956794)
    (10,0.315792566)
}; \label{seller}

\end{axis}

\begin{axis}[
  axis y line*=right,
  yticklabel style={
        /pgf/number format/fixed,
        /pgf/number format/precision=2
},
  ytick={0,0.01,0.02,0.03,0.04},
  legend pos=north west,
  legend style={fill=white, nodes={scale=0.8, transform shape}},
  axis x line=none,
  ymin=0, ymax=0.04,
  ylabel=Average Utility of Winning Buyers,
]
\addlegendimage{/pgfplots/refstyle=seller}\addlegendentry{Sellers}
\addplot[mark=*,blue]
  coordinates{
    (2,0.020157724)
    (4,0.005530165)
    (6,0.004171116)
    (8,0.003831985)
    (10,0.004775161)
}; \addlegendentry{Buyers}
\end{axis}

\end{tikzpicture}
\caption{Average utility of winning buyers and sellers.}
\label{fig:totalutil}
\end{figure}


\begin{figure}[t]
\centering
\pgfplotsset{scaled y ticks=false}
 \begin{tikzpicture}
\pgfplotsset{
    xmin=2, xmax=10,
    width=7cm, compat=1.3
}

\begin{axis}[
  ymin=0, ymax=0.6,
  xlabel=Number of Buyers,
  ylabel=Average Utility of Winning Sellers,
  xtick={2,3,4,5,6,7,8,9,10},
  ytick={0,0.1,0.2,0.3,0.4,0.5,0.6},
  ymajorgrids=true,
  xmajorgrids=true,
  grid style=dashed,
  legend cell align={left},
  legend style={fill=white, nodes={scale=0.8, transform shape}},
  every axis plot/.append style={thick},
  /pgf/number format/.cd,
  1000 sep={}
]
\addplot[mark=x,red]
  coordinates{
    (2,0.273609416)
    (4,0.274812524)
    (6,0.284980794)
    (8,0.29956794)
    (10,0.315792566)
}; \addlegendentry{with Deep Learning}
\addplot[mark=*,blue]
  coordinates{
    (2,0.244814444)
    (4,0.239604142)
    (6,0.251322238)
    (8,0.21908414)
    (10,0.274147555)
}
; \addlegendentry{baseline}

\end{axis}

\end{tikzpicture}
\caption{Average utility of winning sellers with and without deep learning (baseline) in the double auction mechanism.}
\label{utildl}
\end{figure}
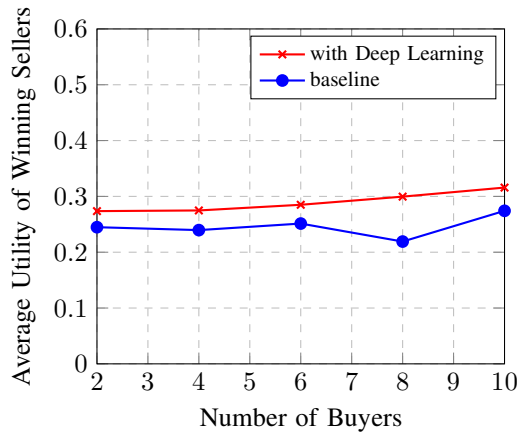

\begin{figure}[t]
\centering
\pgfplotsset{scaled y ticks=false}
 \begin{tikzpicture}
\pgfplotsset{
    xmin=2, xmax=10,
    width=7cm, compat=1.3
}

\begin{axis}[
  ymin=0, ymax=0.7,
  xlabel=Number of Buyers,
  ylabel=Average Utility of Winning Sellers,
  xtick={2,3,4,5,6,7,8,9,10},
  ytick={0,0.1,0.2,0.3,0.4,0.5,0.6,0.7},
  ymajorgrids=true,
  xmajorgrids=true,
  grid style=dashed,
  legend cell align={left},
  legend style={fill=white, nodes={scale=0.8, transform shape}},
  every axis plot/.append style={thick},
  /pgf/number format/.cd,
  1000 sep={}
]
\addplot[mark=x,red]
  coordinates{
    (2,0.296050726)
    (4,0.256871025)
    (6,0.308167955)
    (8,0.337936519)
    (10,0.364629871)
}; \addlegendentry{with Deep Learning, $\theta_m \geq 0.5$}

\addplot[dashed,red]
  coordinates{
    (2,0.100113574)
    (4,0.217876129)
    (6,0.152168025)
    (8,0.15119952)
    (10,0.154217376)
}
; \addlegendentry{with Deep Learning, $\theta_m < 0.5$}

\addplot[mark=*,blue]
  coordinates{
    (2,0.252222063)
    (4,0.220542363)
    (6,0.268727023)
    (8,0.248645835)
    (10,0.315370008)
}
; \addlegendentry{without Deep Learning, $\theta_m \geq 0.5$}

\addplot[mark=*,blue,dashed]
  coordinates{
    (2,0.039493662)
    (4,0.135004628)
    (6,0.156426408)
    (8,0.124685486)
    (10,0.117285937)
}
; \addlegendentry{without Deep Learning, $\theta_m < 0.5$}

\end{axis}

\end{tikzpicture}
\caption{Average utility of winning sellers with different ranges of $\theta_m$.}
\label{utiltheta}
\end{figure}
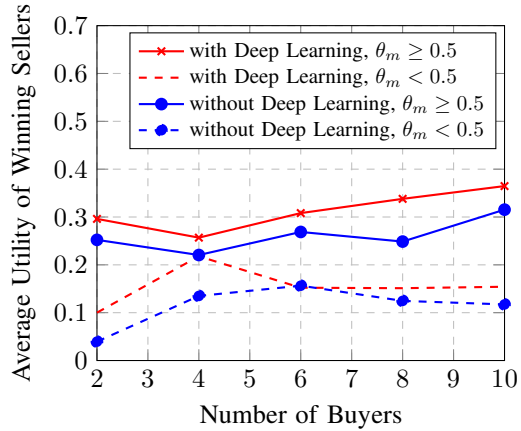

\begin{figure}[t]
\subfigure[]{
\centering
\pgfplotsset{scaled y ticks=false}
 \begin{tikzpicture}
\pgfplotsset{
    xmin=2, xmax=10,
    width=4cm, compat=1.17
}

\begin{axis}[
  label style={font=\scriptsize},
  tick label style={font=\tiny},
  yticklabel style={
        /pgf/number format/fixed,
        /pgf/number format/precision=2},
  ymin=0, ymax=0.6,
  xlabel=Number of Buyers,
  ylabel=Similarity Score,
  xtick={2,3,4,5,6,7,8,9,10},
  ytick={0,0.1,0.2,0.3,0.4,0.5,0.6},
  ymajorgrids=true,
  xmajorgrids=true,
  grid style=dashed,
  legend cell align={left},
  legend style={fill=white, nodes={scale=0.5, transform shape}},
  every axis plot/.append style={thick},
  /pgf/number format/.cd,
  1000 sep={}
]
\addplot[mark=x,red]
  coordinates{
    (2,0.31661977563455185)
    (4,0.37063285257855044)
    (6,0.326431110125172)
    (8,0.3337682992794319)
    (10,0.3457011779589282)
}; \addlegendentry{$\theta_m \geq 0.5$}
\addplot[mark=*,blue]
  coordinates{
    (2,0.1669973423216254)
    (4,0.1426884590199269)
    (6,0.1688055098946527)
    (8,0.1443021703427639)
    (10,0.1588024486502578)
}
; \addlegendentry{$\theta_m < 0.5$}
\end{axis}
\end{tikzpicture}
} 
\subfigure[]{
\centering
\pgfplotsset{scaled y ticks=false}
 \begin{tikzpicture}
\pgfplotsset{
    xmin=2, xmax=10,
    width=4cm, compat=1.17
}

\begin{axis}[
  label style={font=\scriptsize},
  tick label style={font=\tiny},
  yticklabel style={
        /pgf/number format/fixed,
        /pgf/number format/precision=2},
  ymin=0, ymax=0.6,
  xlabel=Number of Buyers,
  ylabel=BLEU Score,
  xtick={2,3,4,5,6,7,8,9,10},
  ytick={0,0.1,0.2,0.3,0.4,0.5,0.6},
  ymajorgrids=true,
  xmajorgrids=true,
  grid style=dashed,
  legend cell align={left},
  legend style={fill=white, nodes={scale=0.5, transform shape}},
  every axis plot/.append style={thick},
  /pgf/number format/.cd,
  1000 sep={}
]
\addplot[mark=x,red]
  coordinates{
    (2,0.32052125495177236)
    (4,0.3791832170253799)
    (6,0.33069722450874767)
    (8,0.33690080743964856)
    (10,0.3474632178279614)
}; \addlegendentry{$\theta_m \geq 0.5$}
\addplot[mark=*,blue]
  coordinates{
    (2,0.16595875755783013)
    (4,0.14013916301765394)
    (6,0.1683687743248769)
    (8,0.14187539916126404)
    (10,0.15628916855205552)
}
; \addlegendentry{$\theta_m < 0.5$}

\end{axis}

\end{tikzpicture}
}
\caption{(a) Average similarity score of sellers (b) Average BLEU score of sellers with different cost of the semantic model, $\theta_m$.}
\label{fig:bleusim}
\end{figure}
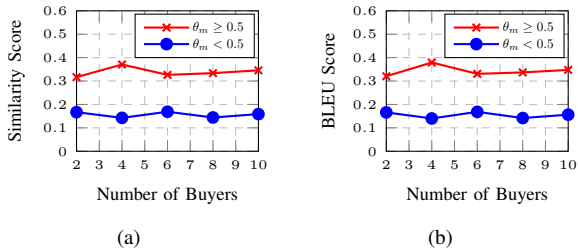

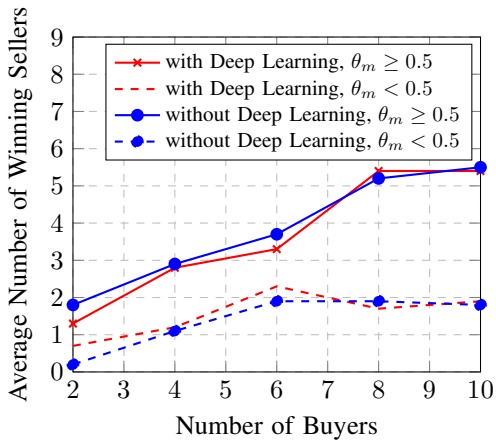
\begin{figure}[t]
\centering
\pgfplotsset{scaled y ticks=false}
 \begin{tikzpicture}
\pgfplotsset{
    xmin=2, xmax=10,
    width=7cm, compat=1.3
}

\begin{axis}[
  ymin=0, ymax=9,
  xlabel=Number of Buyers,
  ylabel=Average Number of Winning Sellers,
  xtick={2,3,4,5,6,7,8,9,10},
  ytick={0,1,2,3,4,5,6,7,8,9},
  ymajorgrids=true,
  xmajorgrids=true,
  grid style=dashed,
  legend cell align={left},
  legend style={fill=white, nodes={scale=0.8, transform shape}},
  every axis plot/.append style={thick},
  /pgf/number format/.cd,
  1000 sep={}
]
\addplot[mark=x,red]
  coordinates{
    (2,1.3)
    (4,2.8)
    (6,3.3)
    (8,5.4)
    (10,5.4)
}; \addlegendentry{with Deep Learning, $\theta_m \geq 0.5$}

\addplot[dashed,red]
  coordinates{
    (2,0.7)
    (4,1.2)
    (6,2.3)
    (8,1.7)
    (10,1.9)
}
; \addlegendentry{with Deep Learning, $\theta_m < 0.5$}

\addplot[mark=*,blue]
  coordinates{
    (2,1.8)
    (4,2.9)
    (6,3.7)
    (8,5.2)
    (10,5.5)
}
; \addlegendentry{without Deep Learning, $\theta_m \geq 0.5$}

\addplot[mark=*,blue,dashed]
  coordinates{
    (2,0.2)
    (4,1.1)
    (6,1.9)
    (8,1.9)
    (10,1.8)
}
; \addlegendentry{without Deep Learning, $\theta_m < 0.5$}

\end{axis}

\end{tikzpicture}
\caption{Average number of winning sellers with and without deep learning in the double auction mechanism.}
\label{fig:noseller}
\end{figure}
\begin{figure}[t]
\subfigure[]{
\centering
\pgfplotsset{scaled y ticks=false}
 \begin{tikzpicture}
\pgfplotsset{
    xmin=0.1, xmax=1,
    width=4cm, compat=1.17
}

\begin{axis}[
  label style={font=\scriptsize},
  tick label style={font=\tiny},
  yticklabel style={
        /pgf/number format/fixed,
        /pgf/number format/precision=2},
  ymin=-0.2, ymax=0.2,
  xlabel=Ask,
  ylabel=Utility,
  xtick={0.2,0.4,0.6,0.8,1.0},
  ytick={-0.2,-0.1,0,0.1,0.2},
  ymajorgrids=true,
  xmajorgrids=true,
  grid style=dashed,
  legend cell align={left},
  legend style={fill=white, nodes={scale=0.5, transform shape}},
  every axis plot/.append style={thick},
  /pgf/number format/.cd,
  1000 sep={}
]
\addplot[mark=x,blue,dashed]
  coordinates{
    (0.1,0.12)
    (0.19,0.12)
    (0.2,0.12)
    (0.3,0)
    (0.4,0)
    (0.5,0)
    (0.6,0)
    (0.7,0)
    (0.8,0)
    (0.9,0)
    (1.0,0)
}; 
\end{axis}
\end{tikzpicture}
} 
\subfigure[]{
\centering
\pgfplotsset{scaled y ticks=false}
 \begin{tikzpicture}
\pgfplotsset{
    xmin=0.1, xmax=1,
    width=4cm, compat=1.17
}

\begin{axis}[
  label style={font=\scriptsize},
  tick label style={font=\tiny},
  yticklabel style={
        /pgf/number format/fixed,
        /pgf/number format/precision=2},
  ymin=-0.2, ymax=0.2,
  xlabel=Ask,
  ylabel=Utility,
  xtick={0.2,0.4,0.6,0.8,1.0},
  ytick={-0.2,-0.1,0,0.1,0.2},
  ymajorgrids=true,
  xmajorgrids=true,
  grid style=dashed,
  legend cell align={left},
  legend style={fill=white, nodes={scale=0.8, transform shape}},
  every axis plot/.append style={thick},
  /pgf/number format/.cd,
  1000 sep={}
]
\addplot[mark=x,blue,dashed]
  coordinates{
    (0.1,0)
    (0.18,0)
    (0.2,0)
    (0.3,0)
    (0.4,0)
    (0.5,0)
    (0.6,0)
    (0.7,0)
    (0.8,0)
    (0.9,0)
    (1.0,0)
}; 

\end{axis}

\end{tikzpicture}
}
\caption{Utility when asking untruthfully by (a) seller $m$ that wins the auction, $m \in \mathcal{M}_w$ (b) seller $m$ that loses the auction, $m \notin \mathcal{M}_w$.}
\label{fig:truthseller}
\end{figure}
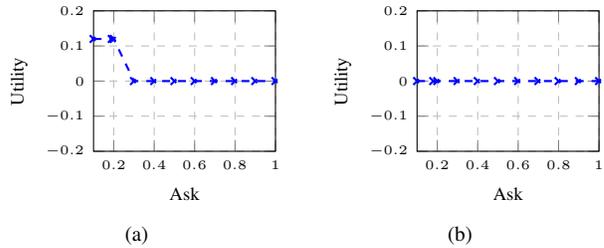

To validate that the double auction mechanism is individually rational and budget balanced, we record the values of ask, bid, and price in one of the samples with $M=20$ and $N=10$. The values are shown in Fig. \ref{bidaskprice}. We observe that there are totally 7 winning seller-buyer pairs, and the utilities for all of the winning pairs are positive. This means that the winning sellers are paid higher than their cost, and the winning buyers pay no more than their true valuation for the semantic information. Therefore, both buyers and sellers have incentives to participate in the auction. For the losing sellers and buyers, their utilities are zero. This shows that the property of individual rationality is achieved because all of the buyers and sellers are awarded with a non-negative utility. The price paid by winning sellers is equal to the payment received by the winning buyers. Thus, the budget balanced property is satisfied.

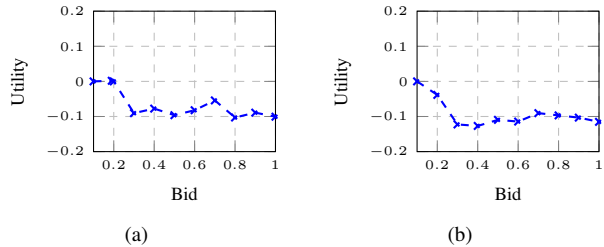
\begin{figure}[t]
\subfigure[]{
\centering
\pgfplotsset{scaled y ticks=false}
 \begin{tikzpicture}
\pgfplotsset{
    xmin=0.1, xmax=1,
    width=4cm, compat=1.17
}

\begin{axis}[
  label style={font=\scriptsize},
  tick label style={font=\tiny},
  yticklabel style={
        /pgf/number format/fixed,
        /pgf/number format/precision=2},
  ymin=-0.2, ymax=0.2,
  xlabel=Bid,
  ylabel=Utility,
  xtick={0.2,0.4,0.6,0.8,1.0},
  ytick={-0.2,-0.1,0,0.1,0.2},
  ymajorgrids=true,
  xmajorgrids=true,
  grid style=dashed,
  legend cell align={left},
  legend style={fill=white, nodes={scale=0.5, transform shape}},
  every axis plot/.append style={thick},
  /pgf/number format/.cd,
  1000 sep={}
]
\addplot[mark=x,blue,dashed]
  coordinates{
    (0.1,0)
    (0.19,0.00177)
    (0.2,0)
    (0.3,-0.09045307895236071)
    (0.4,-0.0775076057582384)
    (0.5,-0.09628250261836108)
    (0.6,-0.08214270135455187)
    (0.7,-0.05417460760169085)
    (0.8,-0.10288744231753405)
    (0.9,-0.08905493279986437)
    (1.0,-0.10054930111460741)
}; 
\end{axis}
\end{tikzpicture}
} 
\subfigure[]{
\centering
\pgfplotsset{scaled y ticks=false}
 \begin{tikzpicture}
\pgfplotsset{
    xmin=0.1, xmax=1,
    width=4cm, compat=1.17
}

\begin{axis}[
  label style={font=\scriptsize},
  tick label style={font=\tiny},
  yticklabel style={
        /pgf/number format/fixed,
        /pgf/number format/precision=2},
  ymin=-0.2, ymax=0.2,
  xlabel=Bid,
  ylabel=Utility,
  xtick={0.2,0.4,0.6,0.8,1.0},
  ytick={-0.2,-0.1,0,0.1,0.2},
  ymajorgrids=true,
  xmajorgrids=true,
  grid style=dashed,
  legend cell align={left},
  legend style={fill=white, nodes={scale=0.8, transform shape}},
  every axis plot/.append style={thick},
  /pgf/number format/.cd,
  1000 sep={}
]
\addplot[mark=x,blue,dashed]
  coordinates{
    (0.1,0)
    (0.2,-0.038536677140890735)
    (0.3,-0.12262935378044762)
    (0.4,-0.1268948707768313)
    (0.5,-0.10954130627602257)
    (0.6,-0.11443690397232689)
    (0.7,-0.09042409994095482)
    (0.8,-0.09685430862396874)
    (0.9,-0.10318235494583763)
    (1.0,-0.1152497027108065)
}; 

\end{axis}

\end{tikzpicture}
}
\caption{Utility when bidding untruthfully by (a) buyer $n$ that wins seller $m$ in the auction, $n \in \mathcal{N}_w$ (b) buyer $n$ that loses seller $m$ in the auction, $n \notin \mathcal{N}_w$.}
\label{fig:truthbuyer}
\end{figure}

The average utility of the winning buyers and sellers are presented in Fig. \ref{fig:totalutil}, which is obtained by averaging the values of 1000 samples. Intuitively, as the number of buyers increases, the sellers have more choices to achieve higher utilities. From Fig. \ref{fig:totalutil}, we observe that the auction mechanism helps to increase the average utility of the winning sellers as the number of buyers grows. Thus, our proposed mechanism can attract more sellers to participate in the information exchange with semantic communication systems.

To investigate the impact of DL in the double auction, we compare the average utilities of winning sellers with and without DL mechanism. The results without the DL mechanism (i.e., the baseline) are obtained by using the double auction mechanism proposed in \cite{jin2015auction}. It is shown in Fig. \ref{utildl} that the average utility of the winning sellers is higher when DL mechanism is adopted in the double auction. The reason is that the DL mechanism helps to maximize the revenue of the sellers.

As shown in Fig. \ref{utiltheta}, the average utility of the winning sellers is higher when the sellers set the payment price $\theta_m \geq 0.5$. The reason is that the sellers with higher $\theta_m$ obtain the semantic model which has a higher BLEU score and similarity score. Hence the buyers are willing to pay more to obtain more accurate information. This insight is verified in Fig. \ref{fig:bleusim}, in which we can see that the similarity score and the BLEU score are higher for sellers with $\theta_m \geq 0.5$. The higher similarity and BLEU scores motivate the buyers to submit higher bids to the sellers, which results in higher utilities as shown in Fig. \ref{utiltheta}. Furthermore, it is shown in Fig. \ref{fig:noseller} that there are more sellers with $\theta_m \geq 0.5$ from the winning sellers. In other words, the seller with higher $\theta_m$ has a higher chance to win the auction, regardless of the number of buyers.

To verify the truthfulness of the double auction, the sellers and buyers are randomly chosen to evaluate their utilities when their bid and ask are different from their true valuation. In Fig. \ref{fig:truthseller} (a), seller $m$ wins and gains the utility $u^s_m = 0.12$ when it asks truthfully with $a_m = C_m = 0.19$. It is shown that the utility cannot be improved by other values of ask. From Fig. \ref{fig:truthseller} (b), seller $m$ loses the auction with truthful ask $u^s_m = 0.18$ obtaining zero utility. It is shown that seller $m$ does not obtain a higher utility when asking untruthfully. In Fig. \ref{fig:truthbuyer} (a), buyer $n$ wins seller $m$ when it bids truthfully with $b^m_n = v^m_n = 0.19$ achieving a utility $u^b_{n,m}=0.0018$. There is no other higher utility achieved when it bids untruthfully. Fig. \ref{fig:truthbuyer} (b) shows the scenario when buyer $n$ does not win seller $m$ and achieve a non-positive utility when it bids untruthfully.

From the experiment results, we observe that the sellers that pay higher prices for the semantic models can achieve better similarity and BLEU scores in the double auction. It is shown that the sellers with better performance are more likely to win the auction and obtain higher utilities. Numerical results also show that the proposed double auction is incentive compatible, individually rational, and budget balanced.

\section{Conclusion And Future Directions}\label{sec:conclusion}
In this paper, we have proposed incentive mechanisms for both semantic model trading and semantic information trading. We developed the valuation functions for general semantic communications, and  performed a case study of the proposed auctions for semantic text transmission. To improve the system performance, we have proposed an effective feature reduction method to support devices with limited transmission resources. Simulation results show that the proposed method helps to increase significantly the utility of devices in the semantic information trading. Moreover, with the double auction mechanism, we have matched  the buyers and devices effectively. It is also shown that the revenue of the semantic model provider can be maximized while keeping the properties of incentive compatibility and individual rationality. For future research directions, we can consider the semantic-aware incentive mechanism design in non-text-based transmission such as wireless images and video transmission, and other semantic-based intelligent tasks.

For future works, considering that the raw data collected from different regions decays over time, we can count the age of information in the value functions of raw data. The difference in the age of information can also be taken into account in the evaluation of transmission accuracy. Moreover, we can consider that the semantic information, which is extracted from different types of raw data, e.g., text, image, and audio, has difference values.


\bibliographystyle{IEEEtran}
\bibliography{refs}

\end{document}